\newif\iffun\funfalse
\documentclass{article}
\usepackage{geometry}
\usepackage{graphicx}
\graphicspath{{./figures/}}
\usepackage{amsmath,amssymb,graphicx,xspace,color,url,nicefrac,microtype}
\usepackage{wrapfig}
\usepackage{booktabs}
\usepackage{enumitem}

\usepackage{amsthm}

\newtheorem{theorem}{Theorem}

\newtheorem{observation}[theorem]{Observation}
\newtheorem{lemma}[theorem]{Lemma}

\newcommand\e\emph\newcommand\eps{\ensuremath{\varepsilon}\xspace}
\renewcommand\P{{\ensuremath P}\xspace}\newcommand\p{{\ensuremath p}\xspace}
\newcommand\s{{\ensuremath s}\xspace}\renewcommand\t{{\ensuremath t}\xspace}
\newcommand\B{{\ensuremath{\mathcal B}}\xspace}\newcommand\T{{\ensuremath{\mathcal T}}\xspace}
\newcommand\spms{\ensuremath{\mathsf{SPM}(\s)}\xspace}\newcommand\spmt{\ensuremath{\mathsf{SPM}(\t)}\xspace}
\newcommand\sps[1]{\ensuremath{\mathsf{SP}(\s,{#1})}\xspace}\newcommand\spt[1]{\ensuremath{\mathsf{SP}(\t,{#1})}\xspace}
\newcommand\rs[1]{\ensuremath{\mathsf r_{\s}({#1})}\xspace}\newcommand\rt[1]{\ensuremath{\mathsf r_{\t}({#1})}\xspace}

\newcommand{\geod}{\ensuremath{\mathsf{SP}}\xspace}
\newcommand{\SPM}{\ensuremath{\mathsf{SPM}}\xspace}

\newcommand\st{barrier\xspace}\newcommand\sts{barriers\xspace}\newcommand\Sts{Barriers\xspace}

\newcommand{\mkmcal}[1]{\ensuremath{\mathcal{#1}}\xspace}

\newcommand{\D}{\mkmcal{D}}

\newcommand{\old}[1]{}

\newcommand{\myremark}[3]{\textcolor{blue}{\textsc{#1 #2:}} \textcolor{red}{\textsf{#3}}}
 \renewcommand{\myremark}[3]{}

\newcommand{\frank}[2][says]{\myremark{Frank}{#1}{#2}}

\newcommand{\maarten}[2][says]{\myremark{Maarten}{#1}{#2}}
\newcommand{\val}[2][says]{\myremark{Val}{#1}{#2}}

\newcommand{\thmheadfont}{\bfseries}
\newenvironment{repeatenv}[2]%
  {\smallskip\noindent {\thmheadfont #1~\ref{#2}.}\ \slshape}
  {\normalfont}

\def\nicefrac#1#2{
	\raise.5ex\hbox{$#1$}%
	\kern-.1em/\kern-.15em%
	\lower.25ex\hbox{$#2$}}

\title{Most vital segment barriers}

\author
{ Irina Kostitsyna\thanks{TU Eindhoven, the Netherlands, {\tt i.kostitsyna@tue.nl}}
  \and Maarten L{\"o}ffler\thanks{Utrecht University, the Netherlands, {\tt m.loffler@uu.nl}}
  \and Valentin Polishchuk\thanks{Link\"oping University, Sweden, {\tt polishchuk@liu.se}}
  \and Frank Staals\thanks{Utrecht University, the Netherlands, {\tt f.staals@uu.nl}}
}

\date{}

\begin{document}






\maketitle
\vspace{-1em}
\begin{abstract}
  We study continuous analogues of ``vitality'' for discrete network
  flows/paths, and consider problems related to placing segment \sts that
  have highest impact on a flow/path in a polygonal domain. This extends the
  graph-theoretic notion of ``most vital arcs'' for flows/paths to geometric
  environments. We give hardness results and efficient algorithms for various
  versions of the problem, (almost) completely separating hard and
  polynomially-solvable cases.
\end{abstract}

\maarten {We just went through the terminology. We now use the macro \st everywhere, and use the term \st instead of stick, blocker, segment, or obstacle where appropriate. The words "segment" and "obstacle" are still used in the text, but with different meanings.}

\section{Introduction}

This paper addresses the following kind of questions:
\begin{quote}{\sl Given a polygonal domain with an ``entry'' and an ``exit'',
    where should one place a given set of ``\sts'' 
    so
    as to decrease the maximum entry-exit flow as much as possible (``flow'' version), or
    to increase the length of the shortest entry-exit
    path as much as possible (``path'' version)?}
\end{quote}
Figure~\ref {fig:intro}
illustrates these questions in their simplest form (placing a single \st in
a simple polygon).  We call the solutions to the problems \e{most vital}
segment \sts for the flow and the path resp.  The name derives from the notion of
\e{most vital arcs} in a network -- those whose deletion decreases the flow or
increases the length of the shortest path as much as possible. While the graph
problems are well studied
\cite{ahuja,alderson2013sometimes,BALL198973,bazgan2015refined,lin1994finding,lubore1971determining,nardelli2001faster,ratliff1975finding},
to our knowledge, geometric versions of locating ``most vital'' facilities have
not been
explored. 
Throughout the paper, the segment \sts will be called simply \sts.
When several segments are aligned to form a longer \st, we
call this longer segment a \e{super-\st}.
We focus
only on segment \sts because already with segments there are a number of
interesting problem versions, and in principle, any polygonal \st may be
created from sufficiently many segments; however, our results imply that the
optimal blocking is always attained by gluing the \sts into super-\sts
(no other configuration of segments is most vital).

Determining the most vital \sts is related to resilience and critical infrastructure protection, as it identifies the most vulnerable spots (``bottlenecks, weakest links'') in the environment by quantifying how fragile or robust the flow/path is, how much it can be hurt, in the worst case, due to an adversarial act. It is thus an example of optimizing from an \e{adversarial} point of view: do as much harm as possible using available budget. 
In practice, the abstract ``bad'' and ``good'' may swap places, e.g., when the
``good guys'' build a defense wall, under constrained resources, to make the
``evil'' (epidemics, enemy, predator, flood) reach a treasure as late as
possible (for the path version) or in a small amount (for the flow
version). Our problem may also be viewed as a Stackelberg game (in networks/graphs parlance aka \e{interdiction problems}  \cite{interdictionReport,interdictionOld,interdictionBandit,interfictionInHandbook,interdiction17}, extensively studied 
due to its relation to security) where the leader places the blockers and the follower computes the maximum flow or the shortest path around them.

Our paper also contributes to the plethora of work on uncertain environments \cite{CitovskyMM17,l-estip-11,Papadimitriou1989}. Motion planning under uncertainty is important, e.g., in computing aircraft paths: locations of hazardous storm systems and other no-fly zones are not known precisely in advance, and it is of interest to understand how much the path or the whole traffic flow may be hurt, in the worst case, if new obstacles pop up (of course, there are many other ways to model weather uncertainty).

Finally, similar types of problems arise when \sts are installed for managing the queue to an airline check-in desk or controlling the flow of spectators to an event entrance.

\begin{figure}[tb]
  \centering
  \includegraphics{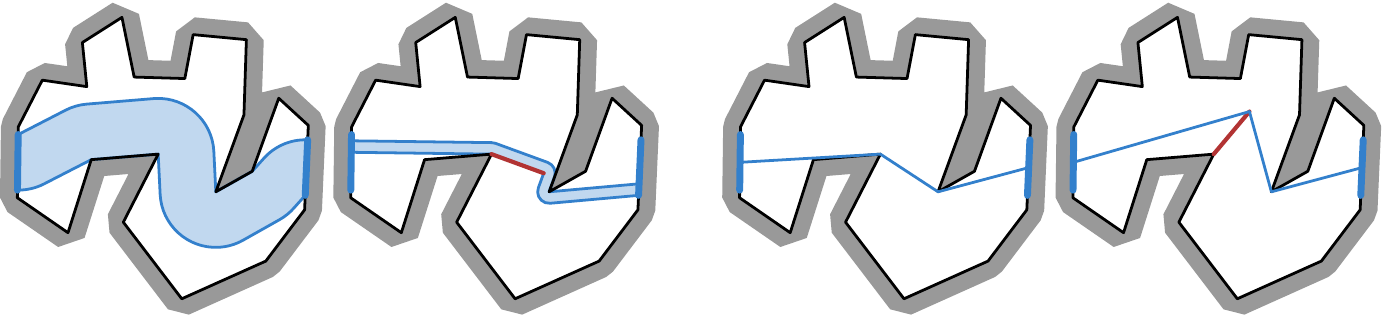}
  \caption{A polygon in which a single \st is placed to minimize the flow between two edges of the polygon (left) or lengthen the shortest path between two points (right). }
  \label{fig:intro}
\end{figure}


\paragraph{Taxonomy.}\label{sec:tax}
Since our input consists of the domain and the \sts, several problem versions may be defined:
\begin{description}
\item[H/h]The domain may have an arbitrary number of holes (such versions will be denoted by~H) or a constant number of holes (denoted by h)
\item[B/b]There may be arbitrarily many \sts (denoted B) or $O(1)$ \sts in the input (denoted b)
\item[D/1]The \sts may have different lengths (denoted D) or all have the same unit length
\end{description}
Overall, for each of the two problems---flow blocking and path blocking---we have 8 versions (HBD, HB1, HbD, Hb1, hBD, hB1, hbD, hb1); e.g., flow-hBD is the problem of blocking the flow in a polygonal domain with $O(1)$ holes using arbitrarily many \sts of different lengths, etc.
We allow \sts to intersect the holes. Depending on the nature of the \sts and the environment, in some of the envisioned applications these may be impractical (e.g., if a hole is pillar in the building, a \st cannot run through it) while in others the assumptions are natural (e.g., if a hole is a pond near the entrance to an event). From the theoretical point of view, in most of our problems these assumptions are w.l.o.g.\ because in the optimal solution the \sts just touch the holes, not ``wasting'' their length inside a hole (one exception is HBD in which the solution may change if the \sts must avoid the holes). 

\paragraph{Overview of the results.} Section~\ref{sec:main} describes our main
technical contribution: a linear-time algorithm for the fundamental problem of
finding \e{one} most vital \st for the shortest \s-\t path in a \e{simple}
polygon. The algorithm is based on observing that the \st must be ``rooted'' at
a vertex of the polygon. The main challenge is thus to trace the locations of
the \st's ``free'' endpoint (the one not touching the polygon boundary) through
the overlay of shortest path maps from \s and~\t. The overlay has quadratic
complexity, so instead of building it, we show that only a linear number of the
maps' cells can be intersected and work out an efficient way to go through all
the cells. Furthermore, we prove that when placing multiple \sts they can be
lined up into a single super-\st; this reduces the problem to that of placing
one \st.
In the remainder of the paper we consider polygons with holes.
Section~\ref{sec:hard} shows hardness of the most general problems flow-HBD and path-HBD, i.e., blocking with multiple different-length \sts in polygons with (a large number of) holes. We also prove weak hardness of the versions with small number of holes (flow-hBD and path-hBD). Finally, we argue that path blocking is weakly hard if the \sts have the same length (path-HB1).
Section~\ref{sec:poly} presents polynomial-time algorithms for 
    path blocking
    with few \sts 
    (path-HbD), implying that path-hbD, path-Hb1 and path-hb1 are also polynomial. 
The section then
describes polynomial-time algorithms for the remaining versions of flow
blocking. We first show that the problem is pseudopolynomial if the \sts have
the same length (flow-HB1). We then prove that blocking with few \sts
(flow-HbD) is strongly polynomial, implying that flow-hbD, flow-Hb1 and
flow-hb1 are also polynomial. Finally, we show polynomiality of the version
with constant number of holes (flow-hB1). Table~\ref{tab:results} summarizes
the hardness and polynomiality of our results.


\begin{table}[tb]
  \centering
\setlength{\tabcolsep}{4pt}
\begin{tabular}{lcccccccc}
  \toprule
       & HBD  & HB1    & HbD    & Hb1  &  hBD & hB1 & hbD  & hb1\\
  \midrule
  Path & NP-hard & weakly NP-hard & poly   & poly & weakly NP-hard &  ? & poly & poly \\
  Flow & NP-hard & pseudo-poly & poly   & poly & weakly NP-hard &  poly & poly & poly \\
  \bottomrule
\end{tabular}
\caption {When the number of holes and \sts exceeds $1$, the problem may become
  (weakly or strongly) NP-hard. This table shows which combinations of
  parameters lead to polynomial or hard problems. The results for $Hb1$, $hbD$
  and $hb1$ follow directly from the result for $HbD$.}
\label{tab:results}
\end{table}

\section{Preliminaries}\label{sec:prelim}


\begin{figure}[b]
\centering
\includegraphics{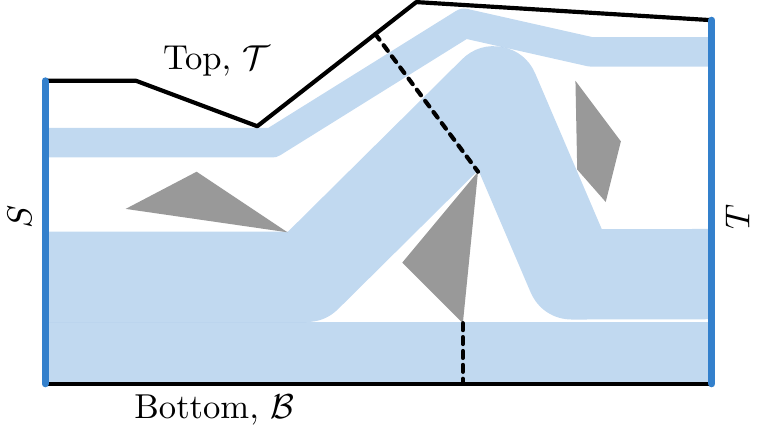}
\caption{Flow setup. An $S$-$T$ flow decomposed into 3 thick paths (yellow); two edges of the critical graph, defining a cut (dashed).}
\label{fig:setup}
\end{figure}

Let \P be a polygonal domain with $n$ vertices, and let the \e{source} $S$ and the \e{sink} $T$ be two given edges on the outer boundary of \P (Fig.~\ref{fig:setup}). 
A \e{flow} in \P is a vector field $F:\P\to\mathbb{R}^2$ with the following properties: $\mathrm{div}\,F(p)=0\,\,\forall p\in\P$ (there are no source/sinks inside the domain), $F(p)\cdot{\bf n}(p)=0\,\,\forall p\in\partial\P\setminus\{S\cup T\}$ where ${\bf n}(p)$ is the unit normal to the boundary of \P at point $p$ (the flow enters/exits \P only through the source/sink), and $|F(p)|\le1\,\forall p\in\P$ (the permeability of any point is 1, i.e., not more than a unit of flow can be pushed through any point -- the flow respects the capacity constraint). Similarly to the discrete network flow, the \e{value} of a continuous flow $F$ is the total flow coming in from the source ($\int_S\,F\!\cdot\!{\bf n}\,\,\mathrm{d}s$) -- since in the interior of \P the flow is divergence-free (flow conserves inside \P), by the divergence theorem, the value is equal to the total flow out of the sink ($-\int_T\,F\!\cdot\!{\bf n}\,\,\mathrm{d}t$). A \e{cut} is a partition of \P into 2 parts with $S,T$ in different parts (analogous to a cut in a network); the \e{capacity} of the cut is the length of the boundary between the parts. Finally, the source and the sink split the outer boundary of \P into two parts called the \e{bottom} \B and the \e{top} \T, and the \e{critical graph} of the domain \cite{gewali} is the complete graph on the domain's holes, \B and \T, whose edge lengths equal to the distances between their endpoints (we assume that the edges are embedded to connect the closest points on the corresponding holes, \B or \T). The celebrated Flow Decomposition and MaxFlow/MinCut theorems for network flows have continuous counterparts: (the support of) a flow decomposes into (thick) paths \cite{socg07}, and the maximum value of the $S$-$T$ flow is equal to the capacity of the minimum cut \cite{strang}; moreover, the mincut is defined by the shortest \B-\T path in the critical graph \cite{mitchell90}.

For shortest path blocking, the setup is a bit more elaborated. Let \s be a
point on the outer boundary of \P, and let $S^*$ be the edge containing \s. We
assume that \s is actually an infinitesimally small gap $s^-s^+$ in the
boundary of \P (with $s^-$ below \s and $s^+$ above), and that the union of the \sts and the holes is not allowed to
contain a path that starts on $S^*$ below $s^-$ and ends on $S^*$ above $s^+$,
completely cutting out \s (Fig.~\ref{fig:door}).\footnote{Other modeling
  choices could have been made; e.g, another way to avoid complete blockage
  could be to introduce a ``protected zone'' around \s \`{a} la in works on \e{geographic mincut}~\cite{alon}. Also a more generic view, outside our scope, could be to combine the flow and path problems into considering \e{minimum-cost} flows \cite{socg07,socg14} (the shortest path is the mincost flow of value~0) and explore how the \sts could influence both the capacity of the domain and the cost of the flow.} W.l.o.g.\ we treat $s^-$ and $s^+$ as vertices of~\P. Similarly, we are given a point \t, modeled as a gap $t^+t^-$ in another edge $T^*$ on the outer boundary of~\P.
\begin{figure}[tb]
\centering
\includegraphics{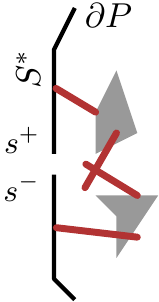}\hfil
\includegraphics[page=1]{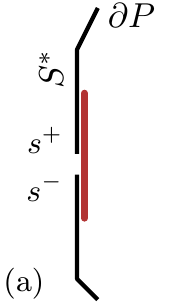}\quad\includegraphics[page=2]{door}\quad\includegraphics[page=3]{door}
\caption{Path setup; \sts are red and \s-\t path is blue. Surrounding $s^-s^+$ (left) is forbidden, even if no \st touches the gap. Completely ``shutting the door'' $s^-s^+$ with one \st (right (a)) is not allowed: if a \st is at \s, it must touch at most one of $s^-,s^+$ (right (b,c)).}\label{fig:door}
\end{figure}

Let $\geod(p,q)$ denote a shortest path (a geodesic) between points $p$ and $q$ in $P$. 
Where it creates no confusion, we will identify a path with its length; in particular, for two points $p,q$, we will use $pq$ to denote both the segment $pq$ and its length. The \e{shortest path map} from \s, denoted \spms, is the decomposition of \P into cells such that shortest paths \sps{\p} from \s to all points \p within a cell visit the same sequence of vertices of \P; 
the last vertex in this sequence is called the \e{root} of the cell and is
denoted by \rs{\p}. The shortest path map from \t (\spmt) and the roots of its
cells (\rt{p}) are defined analogously. The maps have linear complexity and can
be built in $O(n\log n)$ time (in $O(n)$ time if \P is simple)~\cite{SPsurvey}. 
Our algorithm for path blocking in a simple polygon uses:

\begin{lemma}\cite[Lemma~1]{Pollack1989}
  \label{lem:pollack}
  Let $p$, $q$, and $r$ be three points in a simple polygon $P$. The geodesic
  distance from $p$ to a point $x \in \geod(q,r)$ is a convex function of $x$.
\end{lemma}

Finally, let $E(u,v,p)$ denote the ellipse with foci $u$ and $v$, going through the point $p$. It is well known that the sum of distances to the foci is constant along the ellipse; for the points outside (resp.\ inside) the ellipse, the sum is larger (resp.\ smaller) than $up+pv$. It is also well known that the tangent to the ellipse at $p$ is perpendicular to the bisector of the angle $upv$ (the light from $u$ reaches $v$ after reflecting from the ellipse at $p$).

\section{Linear-time algorithms for simple polygons}\label{sec:main}

In this section \P is a simple polygon. For a set $X\subset P$, let
$\geod_X(p,q)$ denote the shortest path between points $p,q$ in $\P\setminus X$
(and the length of the path), i.e., the shortest $p\textrm-q$ path avoiding
$X$. We first consider finding the most vital unit \st for the shortest path,
i.e., finding the unit segment $ab$ maximizing $\geod_{ab}(\s,\t)$. For the
path blocking, we (re)define the bottom \B and top \T of \P as the
$t^-\textrm-s^-$ and $s^+\textrm-t^+$ parts of $\partial\P$ resp. (which mimics
the flow setup, replacing the entrance $S$ and exit $T$ with $s^-s^+$ and
$t^-t^+$). We will treat $s^-, s^+, t^-$, and $t^+$ as vertices of $P$.
We then prove that a most vital \st is placed at a vertex of \P
(Section~\ref{sec:atvertices}). We focus on placing the \st at (a vertex of) \B; placing at \T is symmetric. In Section~\ref{sub:blocking} we test whether
it is possible for any unit \st $ab$ touching \B to also touch \T (while
not lying on $S^*$ or $T^*$): if this is possible, the \st separates \s from \t
completely and $\geod_{ab}(s,t)=\infty$. We test this by computing the
Minkowski sum of \B with a unit disk \val[wonders]{would it suffice to just take the unit disks centered on vertices of \B? would make things little easier (no need for Chin, Snoeyink and Wang e.g.)}
and intersecting the resulting shape with \T, taking special care around \s and \t
(to disallow having $ab\subset S^*$). In Section~\ref{sub:placing} we then
proceed to our main technical contribution: showing how to optimally place a
\st touching (a vertex of) \B given that no such \st can simultaneously
touch \T. For this, we compute the shortest \s-\t path $H$ around the Minkowski
sum of \B with the unit disk and argue that an optimal \st will have one
endpoint on (a vertex of) \B and the other endpoint on $H$. Furthermore, we
show that this path $H$ intersects edges of the shortest path maps \spms and \spmt only
linearly many times. We subdivide $H$ at these intersection points, and show
that for each edge $e$ of $H$ we can then calculate the optimal placement of a
point on $e$ maximizing the sum of distances to \s and~\t. This gives us a
linear-time algorithm for finding a single most vital
\st. In Section~\ref{sec:many_simple} we then show that even if we have
multiple \sts, it is best to glue the \sts
together into a single super-\st.


\subsection {A most vital \st is ``rooted'' at a vertex of \P}\label{sec:atvertices}

We first make some observations about potentially optimal placements. 

\begin{lemma}
  \label{lem:touch}
  A most vital \st touches $\partial\P$.
\end{lemma}
\begin{proof}Suppose that a most vital \st $ab$ does not touch
  $\partial\P$. Clearly, there must be two (equal-length) shortest \s-\t paths, $\pi_a$ and
  $\pi_b$, going through $a$ and $b$ resp., for otherwise the shortest path
  length can be increased by shifting $ab$ along its supporting line. We argue that there is always a direction in which $ab$ can be translated so that the lengths of both $\pi_a$ and $\pi_b$ increase.

\begin{figure}[tb]
\centering\includegraphics{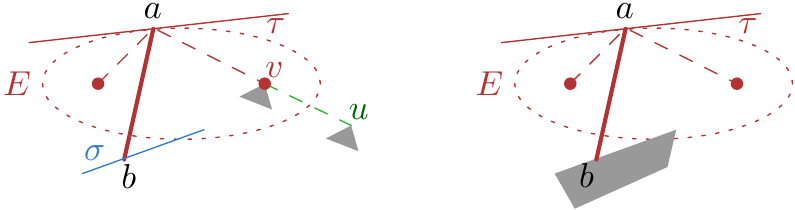}\caption{Red dots are \rs{a} and \rt{a}. Left: It is always possible to translate $ab$ so that $a$ moves into $\tau^+$ and $b$ moves into $\sigma^+$. The tangent will not change when \rt{a} jumps from $v$ to $u$. Right: Sliding $b$ to the right lengthens $\pi_a$}\label{fig:inf}
\end{figure}

  Consider the ellipse $E=E(\rs{a},\rt{a},a)$ through $a$; let $\tau$ be the
  tangent to $E$ at $a$, and let $\tau^+$ be the (closed) halfplane that does
  not contain $E$ (Fig.~\ref{fig:inf}, left). In order to increase
  $\rs{a}a+a\rt{a}$, the \st should be moved so that $a$ moves into
  $\tau^+$. Similarly, let $\sigma^+$ be the closed halfplane, moving $b$ into
  which increases $\rs{b}b+b \rt{b}$ (the halfplane is defined by the tangent
  $\sigma$, at $b$, to the ellipse $E(\rs{b},\rt{b},b)$). There is a direction
  $\theta$ so that the rays in direction $\theta$ starting in $a$ and $b$,
  respectively, are contained in the corresponding half-planes. Hence, we can
  translate $ab$ in this direction so that \e{both} $\rs{a}a+a\rt{a}$ and
  $\rs{b}b+b\rt{b}$ increase. Thus, if none of the roots \rs{a}, \rt{a},
  \rs{b}, \rt{b} changes as the \st is translated, the length of both $\pi_a$
  and $\pi_b$ increases.

It remains to deal with the case in which one of the four roots would change during the infinitesimal translation (say, \rt{a} changes from a vertex $v$ to a vertex $u$) -- i.e., when $a$ belongs to the line $uv$. Let $\tau_v,\tau_u$ be the tangents at $a$ to $E(\rs{a},v,a),E(\rs{a},u,a)$ resp. (refer to Fig.~\ref{fig:inf}, left). Can it be the case that the directions inside $\tau_v^+,\tau_u^+$ and $\sigma^+$ do not have a common direction, i.e., that the good translations defined by $\tau_u$ (the ones increasing $\rs{a}a+au$) are incompatible with those defined by $\tau_v$ (so that $ab$ would be stuck with $a$ on the line $uv$ because the path length would increase both when moving from the cell of $v$ into the cell of $u$ and vice versa)? The answer is no, because $\tau_v=\tau_u$: the former is perpendicular to the bisector of the angle $\rs{a}av$ and the latter is perpendicular to the bisector of the angle $\rs{a}au$ -- which are the same angle.\end{proof}

\begin{lemma}\label{lem:piab}
  A vertex of a most vital \st touches $\partial\P$.
\end{lemma}

\begin{proof}Suppose that none of $a,b$ touches the boundary (so $ab$ touches $\partial\P$ with a point interior to the \st). Clearly, the shortest \s-\t path, $\pi_{ab}$, must go through \e{both} $a$ and $b$, for otherwise the shortest path can be lengthened by moving the \st; also, $\pi_{ab}$ must make same-direction turns (cw or ccw) at both $a$ and $b$, for otherwise the shortest \s-\t path may bypass the \st altogether (Fig.~\ref{fig:piab}, left and middle). We claim that it is always possible to move $ab$ along its supporting line (i.e., keeping the contact with $\partial\P$) increasing the length of $\pi_{ab}$. Indeed, if one of the angles $\rs{a}ab,ab\rt{b}$ is obtuse and the other is acute (Fig.~\ref{fig:piab}, right), then moving in the direction of the acute angle increases both $\rs{a}a$ and $\rt{b}b$ (we assume that the path visits $a$ before $b$). If both angles are acute, then, as can be easily seen by differentiation, the derivative of the path length w.r.t.\ the shift of $ab$ along its supporting line is $\cos(\rs{a}ab))-\cos(ab\rt{b})\ne0$ unless the angles are equal; however, if the angles \e{are} equal, the length is at the minimum (which, again, can be seen by differentiation). The case of both angles being obtuse is similar.\end{proof}
\begin{figure}
\centering
\includegraphics{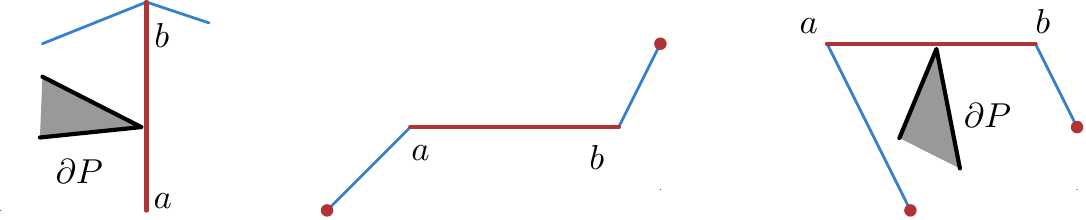}
\caption{The path is blue, red dots are \rs{a} and \rt{b}. Left: If the path bends only on $b$, moving $ab$ up lengthens $\pi_b$. Middle: If the turns are different, $\pi_{ab}$ is not pulled taut. Right: Moving $ab$ left lengthens $\pi_{ab}$}\label{fig:piab}
\end{figure}



\begin{lemma}
  \label{lem:atvertices}
  There exists a most vital barrier $ab$ in which one endpoint, say $b$, lies
  on a vertex of \P.
\end{lemma}
\begin{proof}
Suppose that the \st touches the \e{interior} of an edge of \P (Fig.~\ref{fig:inf}, right). Then the shortest \s-\t path $\pi_a$ through $a$ may be lengthened by translating the \st parallel to itself while sliding $b$ along the edge---the argument is analogous to the one in the proof of Lemma~\ref{lem:touch}: one of the two possible translation directions moves $a$ inside the halfplane~$\tau^+$.
\end{proof}






\subsection{Blocking the path from $s$ to $t$ completely}
\label{sub:blocking}

We now argue that we can check in linear time whether it is possible to completely block passage from $s$ to $t$, by placing a \st that connects \B to \T (without placing the \st along $S^*$ or $T^*$, which is forbidden by our model; see Section~\ref{sec:prelim}).

\begin{observation}
  \label{obs:same_direction_B}
  Let $u$ and $v$ be two vertices of $\geod(s,t)$ in \B. The geodesic makes a
  right turn at $u$ if and only if it makes a right turn at $v$. Let $u'$ and
  $v'$ be two vertices of $\geod(s,t)$ in \T. The geodesic makes a left turn
  at $u'$ if and only if it makes a left turn at $v'$. Moreover, if
  $\geod(s,t)$ makes a right turn in $u$ then it makes a left turn in $u'$.
\end{observation}

Assume without loss of generality that $\geod(s,t)$ makes a right turn at a
vertex $u \in \B$. By Observation~\ref{obs:same_direction_B} it thus makes
right turns at all vertices of $\geod(s,t)\cap \B$, and left turns at all
vertices of $\geod(s,t) \cap \T$.

\begin{observation}
  \label{obs:same-direction}
  If $\geod(s,t)$ makes a right turn at $u \in \B$, and we place a \st
  $ur$ at $u$, then $\geod_{ur}(s,t)$ makes a right
  turn at $r$.
\end{observation}





For every point $p$ on \B, consider placing a \st $pq$ of length at
most one, with one endpoint on $p$. The possible placements $D_p$ of the other endpoint, $q$,
form a subset of the unit disk centered at $p$. Let $\D = \bigcup_{p \in \B} D_p$ denote the union of all these regions (see Fig.~\ref{fig:elipsoid_algo}).

\begin{figure}[tb]
  \centering
  \includegraphics[trim={0 1.5cm 0 1.06cm},clip]{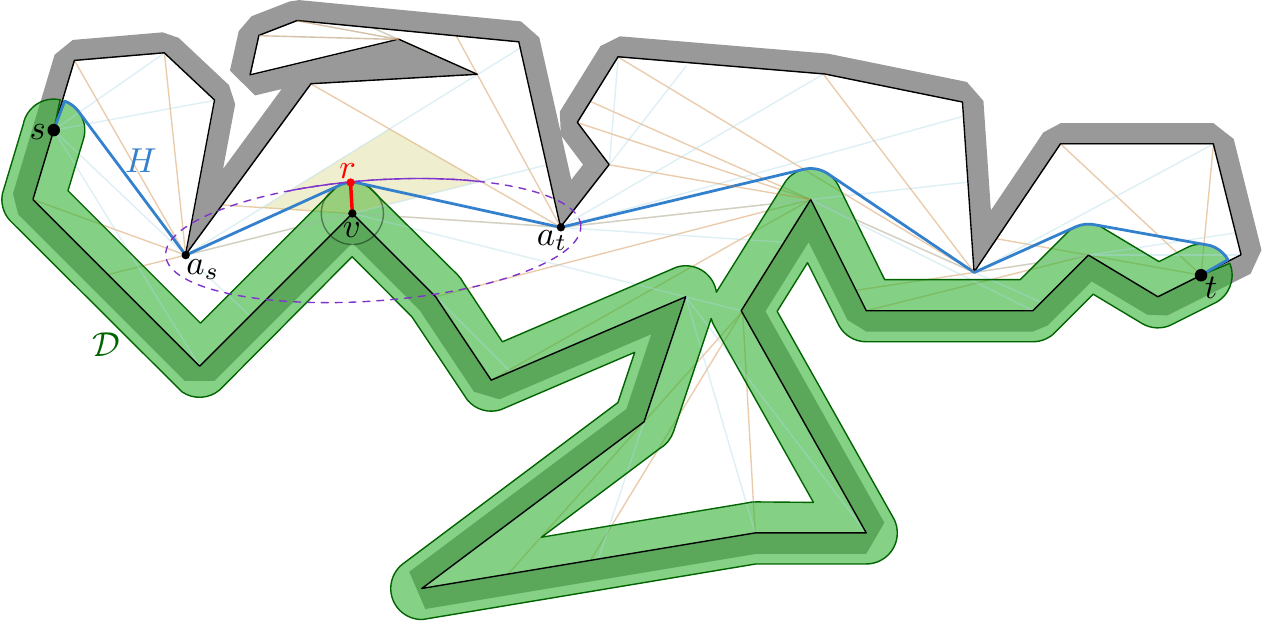}
  \caption{Our algorithm constructs the region \D describing possible
    placements of a \st incident to \B, and the shortest path $H$ around
    \D. An optimal \st incident to \B has one endpoint on $H$. }
  \label{fig:elipsoid_algo}
\end{figure}

\begin{observation}
  \label{obs:separatable}
  There is a \st that separates $s$ from $t$ if and only if $s$ and $t$ are
  in different components of $P \setminus \D$.
\end{observation}

We now observe that \D is essentially the Minkowski sum of \B with a unit disk
$D$. More specifically, let
$A \oplus B = \{ a + b \mid a \in A \land b \in B\}$ denote the Minkowski sum
of $A$ and $B$, let $S^*_\B=S^* \cap \B$ denote the part of $S^*$ in \B, let
$S^*_\T$, $T^*_\B$, and $T^*_\T$ be defined analogously, and let
$\B'= \B \setminus (S^*_\B \cup T^*_\B)$.\val{maybe let's draw a figure (refer to it also from the Lemma)}

\begin{lemma}
  \label{lem:minkowski}
  We have that $\D = \D' \cup X_S \cup X_T$, where $\D' = \B' \oplus D$,
  $X_A = (A^*_\B \oplus D) \setminus A^*_\T$, and $D$ is the unit disk centered at the origin. Moreover,
  \D can be computed in $O(n)$ time.
\end{lemma}

\begin {proof}
  The equality follows directly from the definition of $\D$ and the Minkowski
  sum. It then also follows \D has linear complexity. So we focus on computing
  \D. To this end we separately compute $\D'$, $X_S$, and $X_T$, and take their union. More specifically, we construct the Voronoi diagram of $\B'$ using
  the algorithm of Chin, Snoeyink, and Wang~\cite{csw-fmasp-99}, and use it
  to compute $\B' \oplus D$~\cite{kim1998offset}. Both of these steps can be
  done in linear time. Since $S^*$, $T^*$, and $D$ have constant complexity, we
  can compute $X_S$ and $X_T$ in constant time. The resulting sets still have
  constant complexity, so unioning them with $\B' \oplus D$ takes linear time.
\end {proof}

\begin{lemma}
  \label{lem:compute_P_minD}
  We can test if $s$ and $t$ lie in the same component $C$ of $P \setminus \D$,
  and compute $C$ if it exists, in $O(n)$ time.
\end{lemma}

\begin{proof}
  Using Lemma~\ref{lem:minkowski} we compute \D in linear time. If $s$ or $t$
  lies inside \D, which we can test in linear time, then $C$ does not
  exist. Otherwise, by definition of $X_S$ and $X_T$, $s$ and $t$ must lie on
  the boundary of \D. We then extract the curve $\sigma$ connecting $\s$ to
  $\t$ along the boundary of \D, and test if $\sigma$ intersects the top of the
  polygon \T. If (and only if) $\sigma$ and \T do not intersect, their
  concatenation delineates a single component $C'$ of $P \setminus \D$. Since
  $C'$ contains both $s$ and $t$ we have $C=C'$. So, all that is left is to
  test if $\sigma$ and \T intersect. This can be done in linear time by
  explicitly constructing $C'$ and testing if it is simple~\cite{c-tsplt-91a}.
\end{proof}


\begin{theorem}Given a simple polygon \P with $n$ vertices and two points \s and \t on the boundary of \P, we  can test whether there exists a placement of a unit length \st that  disconnects \s from \t in $O(n)$ time.
\end{theorem}
\val{we don't really need this theorem, let's remove it?}
\frank{I think it makes sense to have some result at the end of this
  subsection. But maybe phrase it as a corollary? Also, I guess we should state
  that s and t lie on the boundary of $P$}
  \val{we work in a specific model (s and t are the "segments", placement on S* is not allowed, etc), so it might be too bold to state a general theorem. I commented out the old statement}
\maarten {I think the old statement is better. It is true that we work in a specific model, but if we would not be doing that, the problem would only become easier (because trivial in some cases) so the statement is not untrue.}

\subsection {Maximizing the length from $s$ to $t$ with a single \st}
\label{sub:placing}

In the remainder of the section we assume that we cannot place a \st on (a
vertex of) \B that completely separates $s$ from $t$. Fix a distance $d$, and
consider all points $p \in P$ such that $\geod(s,p)+\geod(p,t)=d$. Let $C_d$
denote this set of points, and define
$C_{\leq d} = \bigcup_{d' \leq d} C_{d'}$.

Observe that an optimal \st will have one of its endpoints on the boundary of
\D. Let $H = \geod_{\D}(s,t)$ be the shortest path from $s$ to $t$ avoiding
$\D$. We will actually show that there is an optimal \st $V^*$ whose endpoint
$a$ lies on $H$, and that $H$ has low complexity. This then gives us an
efficient algorithm to compute an optimal \st. To show that $a$ lies on $H$ we
use that if $V^*$ realizes detour $d^*$ (i.e., $\geod_{V^*}(s,t)=d^*$), the endpoint $a$ also lies on
$C_{d^*}$. First, we prove some properties of $C_{d^*}$ towards this end.

\begin{observation}
  \label{obs:elliptical}
  Let $\Delta_s$ be a cell in $\SPM(s)$ with root $a_s$, and $\Delta_t$ be a
  cell in $\SPM(t)$ with root $a_t$. We have that
  $C_d \cap \Delta_s \cap \Delta_t$ consists of a constant number of intervals
  along the boundary of the ellipse with foci $a_s$ and $a_t$.
\end{observation}

\begin{proof}
  A point $p \in C_d$ satisfies $\geod(s,p)+\geod(p,t)=d$. For
  $p \in \Delta_s\cap\Delta_t$ we thus have
  $\geod(s,a_s)+\|a_sp\|+\|pa_t\| + \geod(a_t,t) = d$. Since $d$,
  $\geod(s,a_s)$, and $\geod(a_t,t)$ are constant, this equation describes an
  ellipse with foci $a_s$ and $a_t$. Since $\Delta_s$ and $\Delta_t$ have
  constant complexity the lemma follows.
\end{proof}

\begin{lemma}
  \label{lem:geodesically_convex}
  $C_{\leq d}$ is a geodesically convex set (it contains shortest paths between its points).
\end{lemma}

\begin{proof}
  Let $p$ and $q$ be two points on $C_d$, and assume, by contradiction, that
  there is a point $r$ on $\geod(p,q)$ outside of $C_{\leq d}$. By
  Lemma~\ref{lem:pollack} the geodesic distance from $s$ to $\geod(p,q)$ is a
  convex function. Similarly, the distance from $t$ to $\geod(p,q)$ is
  convex. It then follows that the function $f(x) = \geod(s,x)+\geod(x,t)$, for
  $x$ on $\geod(p,q)$ is also convex, and thus has its local maxima at $p$
  and/or $q$. Contradiction. 
\end{proof}



\begin{lemma}
  \label{lem:ray_intersects}
  If there is an optimal \st $ua$ incident to a vertex $u$ of \B, then
  the ray $\rho$ from $u$ through $a$ intersects $H$.
\end{lemma}

\begin{proof}
  The ray $\rho$ splits $P$ into two subpolygons $P_1$ and $P_2$. Since
  $\geod_{ua}(s,t)$ makes a right bend at $a$
  (Observation~\ref{obs:same-direction} and our assumption that $\geod(s,t)$
  makes a right turn at $u$) it intersects both subpolygons $P_1$ and $P_2$. It
  is easy to show that therefore $s$ and $t$ must be in different subpolygons
  (otherwise the geodesic crosses $\rho$ a second time, and we could shortcut
  the path along $\rho$). Since $H$ connects $s$ to $t$ it must thus also
  intersect $\rho$.
\end{proof}

\begin{figure}[tb]
  \centering
  \includegraphics{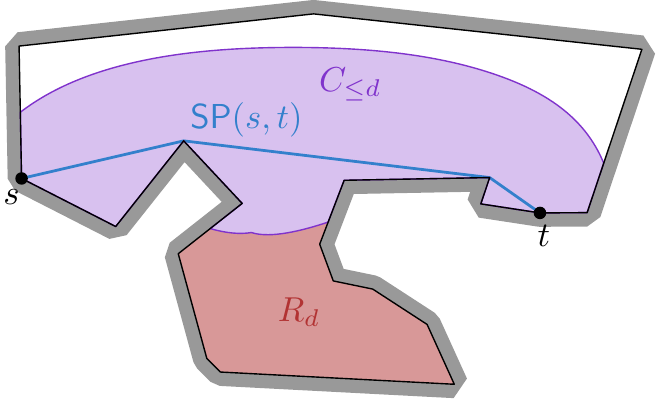}
  \caption{A sketch of the regions $C_{\leq d}$ (purple) and $R_d$. Observe
    that $R_d$ cannot contain any vertices of $\T$, otherwise $\T$ would have
    to pierce $\geod(s,t)$ and thus $C_{\leq d}$.}
  \label{fig:region_R}
\end{figure}
Next, we define the region $R_d$ ``below'' $C_{\leq d}$.
More formally, let $R'$ be the region enclosed by $\B$ and $\geod(s,t)$,
  let $d \geq \geod(s,t)$, and let $R_d = R' \setminus C_{\leq d}$.
  See Fig.~\ref{fig:region_R}.
  We then argue that it is
separated from the top part of our polygon \T, which allows us to prove that
there is an optimal \st with an endpoint on $H$.

\begin{observation}
  \label{obs:bottom_reg_free_of_vertices_of_T}
  Region $R_d$ contains no vertices of \T.
\end{observation}
\begin{proof}
  Assume, by contradiction that there is a vertex of \T in $R_d$. Observe that
  this disconnects $C_{\leq d}$. However, since $C_{\leq d}$ is geodesically
  convex~(Lemma~\ref{lem:geodesically_convex}) and non-empty it is a connected
  set. Contradiction.
\end{proof}

\begin{lemma}
  \label{lem:opt_placement}
  If there is an optimal \st $ua$ where $u$ is a vertex of \B, then there is
  an optimal \st $ur$ where $r$ is a point on $D_u \cap H$ (recall that $D_u$
  is the unit disk centered at~$u$).
\end{lemma}
\begin{proof}
  Assume, by contradiction, that there is no optimal \st incident to $u$ that
  has its other endpoint on $H$.  Consider the ray from $u$ in the direction of
  $a$. By Lemma~\ref{lem:ray_intersects}, the ray hits $H$ in a point $r'$
  (Fig.~\ref{fig:opt_placement}). Because $a$ lies on $C_{d^*}$ and
  $C_{\leq d^*}$ is geodesically convex (Lemma~\ref{lem:geodesically_convex}),
  \maarten {maybe don't need to ref the lemma every time?}
  $r'$ lies outside $C_{\leq d^*}$. Let $H[p,q]=\geod_\D(p,q)$ be the maximal
  (open ended) subpath of $H$ that contains $r'$ and lies outside of
  $C_{\leq d^*}$. We then distinguish two cases, depending on whether or not
  $H[p,q]$ intersects (touches) \D:

  \begin{figure}[b]
    \centering
    \includegraphics{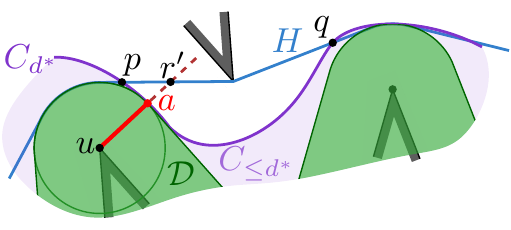}
    \caption{Illustration of Lemma~\ref{lem:opt_placement}.}
    \label{fig:opt_placement}
  \end{figure}

  \begin{description}
  \item[$H{[p,q]}$ does not intersect (touch) \D.] It follows that $H[p,q]$ is a
    geodesic in $P$ as well, i.e. $H[p,q]=\geod(p,q)$. Since
    $p,q \in C_{\leq d^*}$, and $C_{\leq d^*}$ is geodesically convex
    (Lemma~\ref{lem:geodesically_convex}) we then have that
    $H[p,q] \subseteq C_{\leq d^*}$. Contradiction.
  \item[$H{[p,q]}$ intersects $\D$ in a point $z$.] Let $v \in B$ be a point
    such that $z \in D_v$. We distinguish two subcases, depending on whether
    $z$ lies in the region $R_{d^*}$.
    \begin{description}
    \item[$z \in R_{d^*}$.] In this case $z$ lies ``below'' $C_{\leq
        d^*}$. From $z \in H$ it follows that $H[p,q] \subset
      R_{d^*}$. However, as $C_{d^*}$ is geodesically convex, this must mean
      that $H[p,q]$ has a vertex $w$ in $R_{d^*}$ at which it makes a left
      turn. This implies that $w$ is a vertex of \T. By
      Observation~\ref{obs:bottom_reg_free_of_vertices_of_T} there are no vertices of
      \T in $R_{d^*}$. Contradiction.
    \item[$z \not\in R_{d^*}$] Observe that $vz$ is a valid
      candidate \st.  Since $z \not\in C_{\leq d^*}$, the point $z$ actually
      lies above (i.e.~to the left of) $\geod(s,t)$, and thus
      $\geod_{vz}(s,t)$ makes a right turn at $z$. Using that
      $z \not\in C_{\leq d^*}$ it follows that
      $\geod_{vz}(s,t) > d^*$. This contradicts that $d^*$ is the
      maximal detour we can achieve.
    \end{description}
  \end{description}
  Since all cases end in a contradiction this concludes the proof.
\end{proof}


  We now know there exists an optimal \st with an endpoint on $H$. Next, we focus on the complexity of $H$.

\begin{observation}
  \label{obs:left_turns}
  Let $b$ and $c$ be two points on $H$, such that $H$ makes a left turn in
  between $b$ and $c$ (i.e. the subcurve $H[b,c]$ of $H$ between $b$ and $c$
  intersects the half-plane right of the supporting line of
  $bc$). Then $H[b,c]$ contains a vertex of \T.
\end{observation}

\begin{lemma}
  \label{lem:intersect_spm}
  The curve $H$ intersects an edge $e$ of $SPM(z)$, with $z \in \{s,t\}$, at
  most twice. Hence, $H$ intersects $\SPM(z)$ at most $O(n)$ times.
\end{lemma}

\begin{proof}
If $e$ is a polygon edge, then $H$ cannot intersect $e$ at all, so
  consider the case when $e$ is interior to $P$. Assume, by
  contradiction, that $H$ intersects $e$ at least three times, in
  points $a$, $b$, and $c$, in that order along $H$ (Fig.~\ref{fig:intersection_H_spm}).

\begin{figure}[tb]
    \centering
    \includegraphics[page=1]{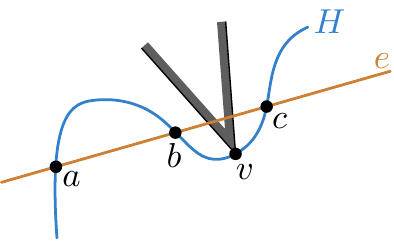}
    \hfil
    \includegraphics[page=2]{intersection_H_spm}
    \caption{The two cases in the proof of Lemma~\ref{lem:intersect_spm}.}
    \label{fig:intersection_H_spm}
  \end{figure}

  If the intersections $a$, $b$, and $c$, are also consecutive along $e$, then
  $H$ makes both a left and right turn in between $a$ and $c$. It is easy to see that since $H$ can bend to the left only at vertices of $\T$
  (Observation~\ref{obs:left_turns}), the region (or one of the two regions) enclosed by $H$ and $ac$ must contain a polygon vertex. Since both $e$ and $H[a,c]$ lie inside $P$, this means that $P$ has a hole. Contradiction.

  If the intersections are not consecutive, (say $a,c,b$), then again there is a region enclosed by $H[a,c]$ and $ab$, containing a polygon vertex. Since both $H[b,c]$ and $cb$ lie inside $P$, this vertex must lie on a hole. Contradiction.
\end{proof}




\paragraph{Algorithm.} We compute intersections of $H$ with the shortest
path maps $\SPM(s)$ and $\SPM(t)$, and subdivide $H$ at each intersection
point. By Lemma~\ref{lem:intersect_spm}, the resulting curve
$H'$ still has only linear complexity. Consider the edges of $H'$ in which $H'$
follows the boundary of $D_v$, for the vertices $v$ of \B. By
Lemma~\ref{lem:opt_placement} for some $v\in\B$ there is an optimal \st that has one
endpoint on such an edge of $H'$ and the other at $v$. Since $H'$ has only $O(n)$
edges we simply try each edge $e$ of $H'$. For all points $r\in e$, the geodesics $\geod(s,r)$ and $\geod(t,r)$ have the same combinatorial structure, i.e., the roots $a_s=\rs{r},a_t=\rt{r}$ stay the same. It follows that we have a constant-size subproblem in which we can compute an optimal \st in constant time.
Specifically, we compute the smallest ellipse $E$ with foci $a_s$ and $a_t$ that contains $e$ and
goes through the point $r$ in which $E$ and $e=D_v$ have a
common tangent (if such a point exists). See Fig.~\ref{fig:elipsoid_algo}. For
that point $r$, we then also know the length of the shortest path
$\geod_{vr}(s,t)=\geod(s,r)+\geod(r,t)$, assuming that we place the \st
$vr$. We then report the point $r$ that maximizes this length over
all edges of $H'$.

Constructing the connected component $P'$ of $P \setminus \D$ that contains $s$
and $t$ takes linear time (Lemma~\ref{lem:compute_P_minD}). This component $P'$
is a simple splinegon, in which we can compute the shortest path $H$ connecting
$s$ to $t$ in $O(n)$ time~\cite{melissartos1992shortest}. Computing $\SPM(s)$
and $\SPM(t)$ also requires linear time~\cite{ghlst-ltavs-87}. We can then walk
along $H$, keeping track of the cells of $\SPM(s)$ and $\SPM(t)$ containing the
current point on $H$.  Computing the ellipse, the point $p$ on the current edge
$e$, and the length of the geodesic takes constant time. It follows that we can
compute an optimal \st incident to \B in linear time. We use the same
procedure to compute an optimal \st incident to \T. We thus obtain the
following result.





\begin{theorem}\label{thm:biggest_stick_simple_polygon}Given a simple polygon $P$ with $n$ vertices and two points $s$ and $t$ on $\partial P$, we can compute a unit length \st that maximizes the length of the shortest path between $s$ and $t$ in $O(n)$ time.\end{theorem}
\val[rephrased the thm]{we work in a specific model (s and t are the "segments", placement on S* is not allowed, etc), so it might be too bold to state a general theorem}
\maarten {Also here, I think the old statement is fine, and I prefer to be more precise.}


\subsection{Using multiple vital \sts}
\label{sec:many_simple}

We prove a structural property that even when we are given many \sts,
there always exists an optimal solution in which they glued into a single super-\st.
This implies that our linear-time algorithm from the previous section can still be used to solve the problem.

Clearly, any solution distributes the \sts over some (unknown) number of super-\sts.
First observe that, similarly to Section~\ref{sec:atvertices}, any super-\st must have a vertex at a vertex of \P
%
%
, and the next lemmas prove that it is
suboptimal to have more than one such super-\st.

\frank{I guess we should still say something about testing if we can block the
  path completely if we have multiple \sts
}

The first lemma is a variant of Lemma~\ref{lem:pollack}.
Let $a_1 b_1$ and $a_2 b_2$ be two segments inside $P$, and let $m_1\in a_1 b_1,m_2\in a_2 b_2$ be two points that divide the segments in the same proportion, that is $\vec{m_1}=\gamma\vec{a_1}+(1-\gamma)\vec{b_1},\vec{m_2}=\gamma\vec{a_2}+(1-\gamma)\vec{b_2}$ for some $\gamma\in[0,1]$. Define function $f(\gamma)=\geod(m_1,m_2)$. 

\begin{lemma}\label{lem:middle-path}
$f(\gamma)$ is convex for $\gamma\in[0,1]$: $\geod(m_1,m_2)\leq \gamma\geod(a_1,a_2)+(1-\gamma)\geod(b_1,b_2)\,.$
\end{lemma}

\begin{proof}
Consider two cases: when the paths $\geod(a_1,a_2),\geod(b_1,b_2)$ intersect and when they do not.

{\bf Case 1.} Paths $\geod(a_1,a_2)$ and $\geod(b_1,b_2)$ intersect (Fig.~\ref{fig:middle-path}, left). Choose some intersection point $c$. From Lemma~\ref{lem:pollack} it follows that a shortest path from a fixed point to a point on a given segment is a convex function. Then,
\[
\geod(c,m_1)\leq \gamma\geod(c,a_1)+(1-\gamma)\geod(c,b_1)\,\quad\text{ and }\quad
\geod(c,m_2)\leq \gamma\geod(c,a_2)+(1-\gamma)\geod(c,b_2)\,.
\]
By triangle inequality,
\[
\geod(m_1,m_2)\leq\geod(c,m_1)+\geod(c,m_2)\,.
\]
Thus,
\[
\geod(m_1,m_2)\leq \gamma\geod(a_1,a_2)+(1-\gamma)\geod(b_1,b_2)\,.
\]

\begin{figure}
  \centering\hfil
  \includegraphics[width=.3\columnwidth,page=3]{middle-path}
  \hfil
  \includegraphics[width=.3\columnwidth,page=1]{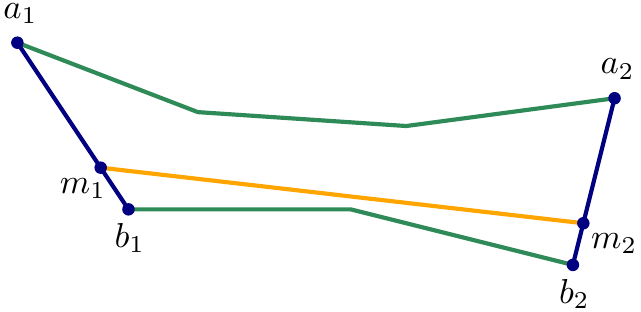}\hfil  \includegraphics[width=.3\columnwidth,page=2]{middle-path}\hfil
  \caption{From left to right: Cases 1, 2a, 2b from Lemma~\ref{lem:middle-path}.}
  \label{fig:middle-path}
\end{figure}

{\bf Case 2.} Paths $\geod(a_1,a_2)$ and $\geod(b_1,b_2)$ do not intersect.

{\bf Case 2a.} Suppose the path $\geod(m_1,m_2)$ is also disjoint from the paths $\geod(a_1,a_2)$ and $\geod(b_1,b_2)$ (Fig.~\ref{fig:middle-path}, middle). Because $P$ is a simple polygon, the path $\geod(m_1,m_2)$ must be a straight-line segment. By simple vector manipulation we can show that
$m_1 m_2\leq \gamma a_1 a_2+(1-\gamma)b_1 b_2$, and thus,
\[
\geod(m_1,m_2)=m_1 m_2\leq \gamma a_1 a_2+(1-\gamma)b_1 b_2 \leq \gamma\geod(a_1,a_2)+(1-\gamma)\geod(b_1,b_2)\,.
\]

{\bf Case 2b.} Finally, w.l.o.g., assume that the path $\geod(m_1,m_2)$ intersects $\geod(b_1,b_2)$, but not $\geod(a_1,a_2)$ (Fig.~\ref{fig:middle-path}, right). Let vertex $b'$ of $P$ lie on both paths $\geod(m_1,m_2)$ and $\geod(b_1,b_2)$. Because $P$ is a simple polygon, there exists a point $a'\in\geod(a_1,a_2)$ visible to $b'$. Let $m'\in a'b'$ divide the segment $a'b'$ in the same proportion, i.e., $\vec{m'}=\gamma\vec{a'}+(1-\gamma)\vec{b'}$. We can apply the same line of reasoning to the segments $a_1 b_1$ and $a'b'$, and the segments $a'b'$ and $a_2 b_2$. Either the case 2a will hold, or we have arrived at the same case 2b but a smaller size instance. Thus, recursively we can show that
\[
\geod(m_1,m')\leq \gamma\geod(a_1,a')+(1-\gamma)\geod(b_1,b')\,,
\]
and
\[
\geod(m',m_2)\leq \gamma\geod(a',a_2)+(1-\gamma)\geod(b',b_2)\,.
\]
By the triangle inequality we have that $\geod(m_1,m_2)\leq\geod(m_1,m')+\geod(m',m_2)$, and thus \val{replaced 1/2s with $\gamma,1-\gamma$}
\[
\geod(m_1,m_2)\leq \gamma\geod(a_1,a_2)+(1-\gamma)\geod(b_1,b_2)\,.\qedhere
\]
%
%
\end{proof}

\begin{lemma}\label{lem:two-sticks}
Given two \sts, possibly of different lengths, an optimal configuration will stack them into a single super-\st.
\end{lemma}
\begin{proof}
  First note that if the two \sts intersect, their endpoints must
  coincide. To see this, treat one of the \sts as (a part of) a hole; then,
  by Lemmas~\ref{lem:pi_atildeb} and~\ref{lem:piab_holes} (which are proved in
  Section~\ref{sec:poly} analogously to Lemmas~\ref{lem:touch}
  and~\ref{lem:piab}) and Lemma~\ref{lem:atvertices}, the \sts must touch
  each other with their vertices. It is easy to see that at the point of the
  touching, the \sts must form a $180^{\circ}$ angle.  \frank{I don't think
    this is true, again due to weirdness around $S^*$;
  } \val{In this case the same
    complete blockage can be attained by two aligned \sts. Let's work under
    the assumption that complete blockage is not possible (then I think the
    claim about 180 deg holds). If needed, we may assume $S^*$
    is longer than all \sts} In what follows we will assume that the
  \sts are disjoint.  \frank{Ah, you are right, it was possible to block
    with an aligned \st. in case $S^*$ is basically a reflex vertex
    (i.e. replace reflex vtx by a small seg) then you do need two
    non-consecutive \sts. Admittedly, in that case you can indeed also cut
    off $s$ from $t$. We say that we can test that with the medial axis
    approach; but even that argument is not fully complete: the medial axis
    gives us a way of testing if we can block with two \sts (i.e. is there a
    point on medial axis at dist $\leq 1$ from \B and \T), but what if you
    would need 3 segs to fully block $s$ from $t$? For that we need that you
    would glue the \sts together again. (Or maybe that is not even possible
    in our current model). Long story short: I believe this will all work out,
    but we should be more specific about it.}

Let the two \sts be the segments $p m_1,q m_2$ where $p,q$ are vertices of \P, and let $\alpha,\beta$ be their lengths. Let point $a$ lie on the extension of segment $p m_1$ into \P with $m_1 a=q m_2=\beta$ (i.e., extend $p m_1$ for distance $\beta$); similarly, let $b$ lie on the extension of $q m_2$ with $m_2 b=p m_1=\alpha$ (Fig.~\ref{fig:two-sticks}). If $m_1 a$ (or $m_2 b$) intersects the polygon boundary, then either $s$ and $t$ get disconnected, or the shortest path from $s$ to $t$ cannot pass through $m_1$ (or $m_2$ resp.). In both cases, the original configuration of the \sts was not optimal.

\begin{figure}
\centering
\includegraphics[width=0.33\textwidth,page=1]{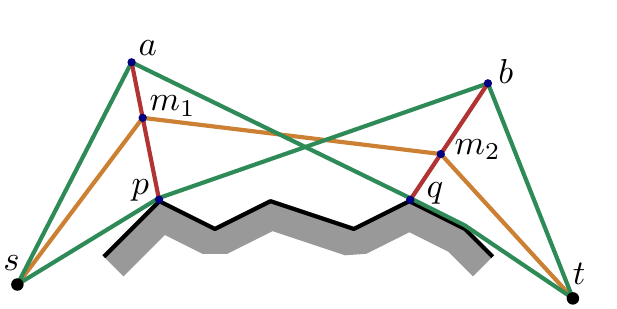}\hfill
\includegraphics[width=0.33\textwidth,page=3]{two-sticks}\hfill
\includegraphics[width=0.33\textwidth,page=2]{two-sticks}
  \caption{Illustration to Lemma~\ref{lem:two-sticks}. The potential \sts are red.}
  \label{fig:two-sticks}
\end{figure}

Then, assume that $m_1 a$ and $m_2 b$ are interior to $P$. W.l.o.g.\ assume that $m_1$ lies between $s$ and $m_2$ on the shortest path from $s$ to $t$. For simplicity assume that the path $\geod(s,b)$ passes through point $p$, and that the path $\geod(a,t)$ passes through point $q$ (Fig.~\ref{fig:two-sticks}, left and middle). Later we will lift this assumption.

Applying Lemma~\ref{lem:pollack} we get that
\[
\geod(s,m_1)\leq \frac{\alpha}{\alpha+\beta}\geod(s,a) + \frac{\beta}{\alpha+\beta}\geod(s,p)\,,
\]
and
\[
\geod(m_2,t)\leq \frac{\beta}{\alpha+\beta}\geod(b,t) + \frac{\alpha}{\alpha+\beta}\geod(q,t)\,.
\]
Applying Lemma~\ref{lem:middle-path} we get that
\[
\geod(m_1,m_2)\leq \frac{\alpha}{\alpha+\beta}\geod(a,q) + \frac{\beta}{\alpha+\beta}\geod(p,b)\,.
\]
Then, summing up these inequalities we get
\begin{multline*}
\geod_{p m_1,q m_2}(s,t) = \geod(s,m_1) + \geod(m_1,m_2) + \geod(m_2,t) \\
\leq \frac{\alpha}{\alpha+\beta}(\geod(s,a)+\geod(a,q)+\geod(q,t)) + \frac{\beta}{\alpha+\beta}(\geod(s,p)+\geod(p,b)+\geod(b,t)) \\
=\frac{\alpha}{\alpha+\beta}\geod_{pa}(s,t) + \frac{\beta}{\alpha+\beta}\geod_{pb}(s,t)\,.
\end{multline*}
Thus, either $\geod_{p m_1,q m_2}(s,t) \leq \geod_{pa}(s,t)$ or $\geod_{p m_1,q m_2}(s,t) \leq \geod_{pb}(s,t)$ (or both).

Now assume that path $\geod(s,b)$ intersects segment $pa$ in point $p'$, and path $\geod(a,t)$ intersects $qb$ in point $q'$ (Fig.~\ref{fig:two-sticks}, right). Then, let points $m'_1$ and $m'_2$ be the points dividing the segments $p'a$ and $q'b$ in the proportion $\alpha:\beta$ and $\beta:\alpha$ respectively. Applying the same argument as above, we get that
\[
\geod_{p m_1,q m_2}(s,t) \leq \frac{\alpha}{\alpha+\beta}\geod_{pa}(s,t) + \frac{\beta}{\alpha+\beta}\geod_{pb}(s,t)\,.
\]
It also must hold that
\[
\geod_{p m'_1,q m'_2}(s,t)\geq\geod_{p m_1,q m_2}(s,t)\,.
\]
We can again conclude that at least one of the following inequalities hold, $\geod_{p m_1,q m_2}(s,t) \leq \geod_{pa}(s,t)$ or $\geod_{p m_1,q m_2}(s,t) \leq \geod_{pb}(s,t)$. That is, the original configuration of the two blocking \sts is not optimal. 
\end{proof}

Lemma~\ref{lem:two-sticks} resolves the question for two \sts that are attached to the boundary of $\P$.
We use induction to extend the result to an arbitrary number of \sts.

\begin{theorem}
\label{thm:many-sticks}
Given a simple polygon $P$ with $n$ vertices, two points $s$ and $t$ on $\partial P$, and $k$ unit-length \sts, the optimal placement of the \sts which maximizes the length of the shortest path between $s$ and $t$ consists of a single super-\st.
\end {theorem}
\begin {proof}
  We use induction on $k$.
  The base case for $k=2$ follows from Lemma~\ref{lem:two-sticks}.

  First, we argue that Lemma~\ref {lem:touch} generalizes to arbitrary rigid configurations of \sts: as we translate the configuration, the length of the shortest path above (below) the configuration is a convex differentiable function in the translation vector, and hence there must still be at least one direction in which the configuration can be translated so as to increase the lengths of both paths.
  Thus, we may assume that the optimal configuration of the $k$ \sts touches $\partial\P$.

  But now, consider (one of) the \st that touches $\partial\P$, and consider it to be part of $\P$.
  By the induction hypothesis, the remaining $k-1$ \sts are combined into a single super-\st. But then, the optimal solution for $k$ \sts must be the same as the optimal solution to a problem in which we have only two \sts: one of length $1$ and one of length $k-1$. We apply Lemma~\ref{lem:two-sticks} again and conclude that the optimal solution uses, in fact, a single super-\st of length $k$.
\end {proof}

Theorem~\ref{thm:many-sticks} implies that 
placing an arbitrary set of \sts reduces to 
placing just one super-\st, so our linear-time algorithm from the previous section applies.

\subsection{Most vital \sts for the flow}\label{sec:flowSimple}In simple polygons the critical graph has only two vertices -- \B and \T (which, for the flow blocking, are the $T\textrm-S$ and $S\textrm-T$ parts of $\partial\P$; refer to Fig.~\ref{fig:setup}). Flow blocking thus boils down to finding the shortest \B-\T connection (then all the \sts will be placed along the connecting segment) -- a problem that was solved in linear time in~\cite{mitchell90}.

\section{Hardness results}\label{sec:hard}

In the remainder of the paper \P is a polygonal domain with holes (as defined in Section~\ref{sec:prelim}).
%
%
We first show that in general it is hard to decide whether full blockage can be achieved, i.e., whether it is possible to decrease the $S$-$T$ flow to 0 or to lengthen the \s-\t path to infinity:

\begin{theorem}
  \label{thm:3part}
  Flow-HBD and path-HBD are NP-hard.
\end{theorem}
\begin{proof}
\begin{figure}
\centering
\includegraphics{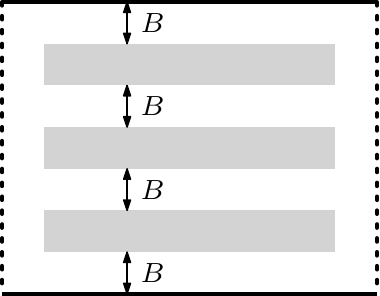}
\caption{Full blockage is possible iff each channel is fully blocked}
\label{fig:3part}
\end{figure}
Figure~\ref{fig:3part} shows the reduction for flow-HBD (for path-HBD, just replace the edges $S,T$ with the points \s,\t): given an instance of 3-Partition (``Can $3m$ given integers $\{a_1\dots a_{3m}\}$ be split into triples so that the sum of integers in each triple is equal to $B=\sum a_i/m$?''), construct the domain with $m$ width-$B$ channels between $S$ and $T$, and have a length-$a_i$ \st (line segment) for each integer $i=1\dots3m$; the \sts can cut $S$ from $T$ iff the 3-Partition instance is feasible.
\end{proof}
With different-length \sts, the full blockage remains (weakly) hard even if $n=O(1)$ (and the number of holes is, of course, also small):
\begin{theorem}
Flow-hBD and path-hBD are weakly NP-hard.
\end{theorem}
\begin{proof}
Given an instance of 2-Partition (``Can $m$ given integers $\{a_1\dots a_{m}\}$ be split into two sets so that the sum of integers in set is equal to $B=\sum a_i/2$?''), construct the domain with $2$ width-$B$ channels between $S$ and $T$ (analogously to the proof of Theorem~\ref{thm:3part}), and have a length-$a_i$ \st for each integer $i=1\dots m$; the \sts can cut $S$ from $T$ iff the 2-Partition instance is feasible.
\end{proof}
If all \sts have the same length, deciding possibility of full blockage is polynomial; in fact, Section~\ref{sec:flow} shows that even the more general flow-HB1 problem (finding how to maximally decrease the flow) can be solved in polynomial time. On the contrary, (partial) path blockage is hard even for unit \sts:
\begin{theorem}
\label{thm:multi_stick_hard}
Path-HB1 is weakly NP-hard.
\end{theorem}
\begin{proof}Given an instance $A=\{a_1\dots a_{m}\}$ of 2-Partition, we construct a domain with $O(m)$ vertices, in which it is possible to place $m$ unit \sts so that the shortest \s-\t path has length $C + \sum(A)/2$ (where $C$ is some number and $\sum(X)$ is the sum of numbers in a set $X$), iff $A$ can be partitioned into two sets $R$ and $B$ with $\sum(R) = \sum(B) = \sum(A)/2$.

  \begin{figure}[tb]
    \centering
    \includegraphics{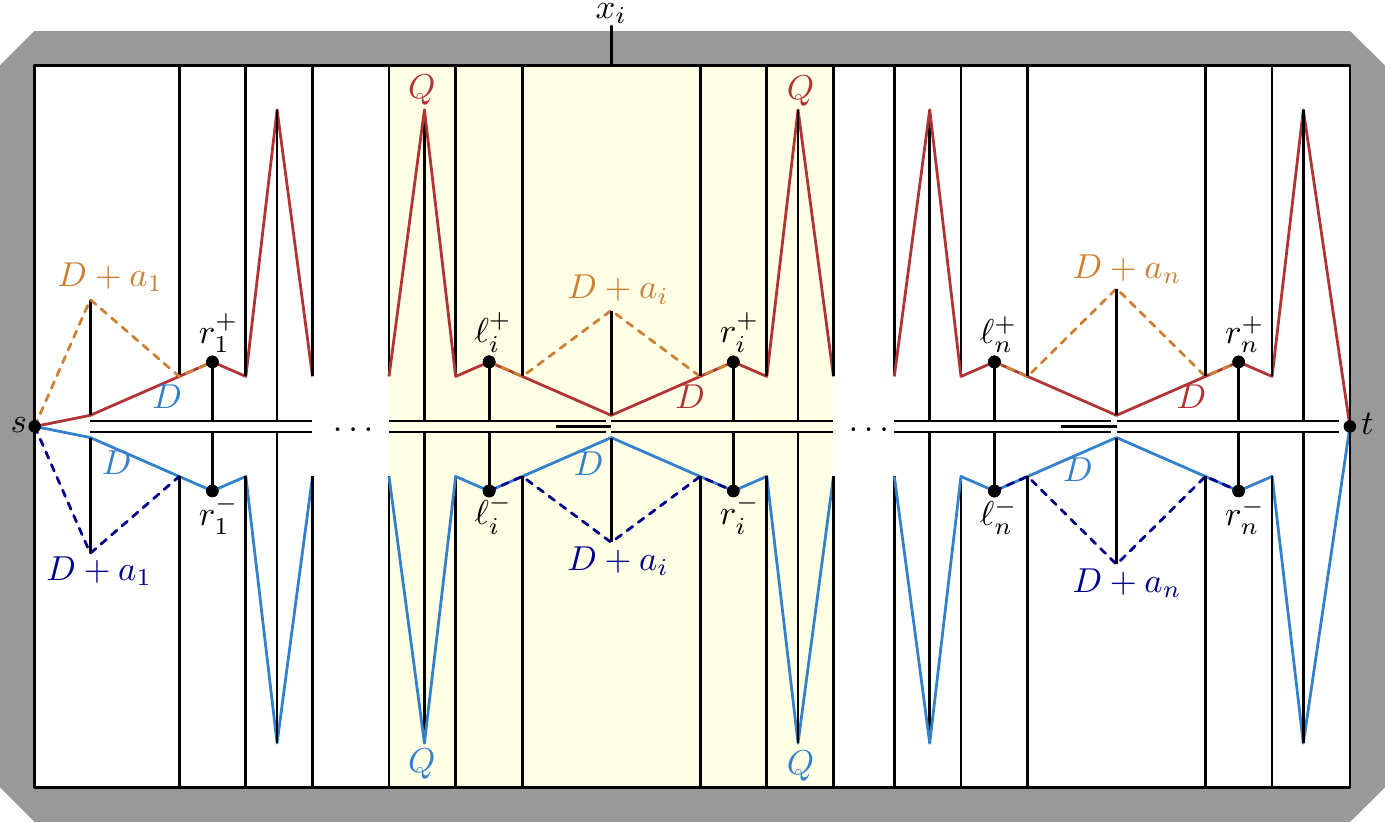}
    \caption{An overview of the construction.}
    \label{fig:hardness_global}
  \end{figure}

  Our construction is sketched in Fig.~\ref{fig:hardness_global}. The main idea
  is that there are three main routes from $s$ to $t$: a very short middle
  route, a ``red'' (top) route, and a ``blue'' (bottom) route. The red and
  blue routes both have length $C=mD+(m-1)E$, for some large constants $D$ and
  $E$, and are much longer than the middle route.  We make it such that with
  exactly $m$ \sts, we can block off the middle route completely. Moreover,
  by placing a \st $i$ appropriately, we can increase the length of either
  the red route or the blue route by exactly $a_i$. Increasing the length of
  the red route by $a_i$ corresponds to assigning $a_i$ to $R$ and increasing
  the length of the blue route by $a_i$ corresponds to assigning $a_i$ to
  $B$. So, in the end the length of the red and blue routes are $C+\sum(R)$ and
  $C+\sum(B)$, respectively. Hence, it is possible to partition $A$ into $R$
  and $B$, with $\sum(R)=\sum(B)=\sum(A)/2$ if and only if the length of the
  shortest path between $s$ and $t$ is at least $C+\sum(A)/2$.

  \paragraph{Description of the Construction.}  In detail, the outer boundary
  of $P$ is a large rectangle, centered vertically at the $x$-axis and \s at the origin. 

  Our middle route will consist of a rectangular corridor of height $1/2$
  vertically centered at the $x$-axis. We cut through the top and bottom walls
  of the corridor in $m$ intervals $[x_i-\eps,x_i]$, for some arbitrarily small
  $\eps>0$. Let $e_i$ and $f_i$ be the points with $x$-coordinate $x_i$ on the
  top and bottom wall, respectively. See Fig.~\ref{fig:hardness_stick}. We
  build another horizontal segment $h_i$ of length $H>2m+\sum(A)$ whose right
  endpoint is $(x_i,0)$.  Note that we have vertices with exactly the same
  $x$-coordinate; also note that in order for the construction to work, we must
  allow placing the \st so that it contains vertices of the
  domain.

  \begin{figure}
    \centering
    \includegraphics{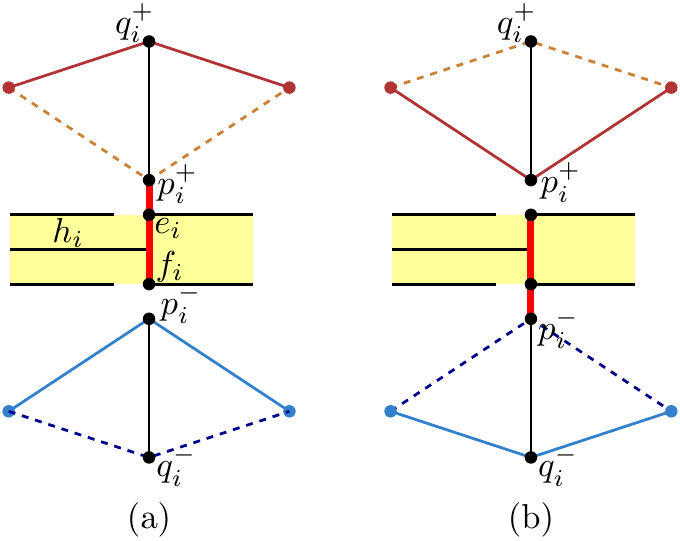}
    \caption{The two possible placements for \st $i$.}
    \label{fig:hardness_stick}
  \end{figure}

  Our red (top) and blue (bottom) routes are completely symmetric, so we
  describe only the top route. For each opening $[x_i-\eps,x_i]$ in the middle
  corridor, we build a vertical segment $p^+_iq^+_i$ with bottom
  endpoint $p^+_i=(x_i,3/4)$. In between every consecutive pair $x_i,x_{i+1}$
  we build three long vertical walls attached to the middle corridor, and three
  long vertical walls attached to the top boundary of $P$ that force the red
  route to zigzag (see Fig.~\ref{fig:hardness_global}). Let $r^+_i$, $k_i$, and
  $\ell^+_{i+1}$ be the top-endpoints of the walls connected to the middle
  corridor. The distances between these walls (and the walls extending from the
  top of $P$) are all large, i.e., larger than $m$, so that we cannot block the
  passage even if we place all \sts consecutively. The walls extending from
  the top of $\partial P$ are built so that the length of the shortest path
  between $r^+_i$ and $\ell^+_{i+1}$ via $k_i$, i.e. our zigzag, is very large,
  say $Q > m^2+\sum(A)$. Let $D>2m$ be the distance from $\ell^+_i$ via $p^+_i$
  to $r^+_i$. We place $q^+_i$ (and the endpoints of the walls extending from
  the top) such that the distance from $\ell^+_i$ via $q^+_i$ to $r^+_i$ is
  $D+a_i$.

  \paragraph{Making sure that we use $m$ \sts to block the middle route.}
  Note that for the shortest \s-\t path to have length at least
  $C=m(Q+D) > m^3+m\sum(A)$, it needs to pass through at least $m$ zigzags from
  our red or blue routes. Hence, we need to block the middle route between any
  pair of consecutive openings $x_i$ and $x_{i+1}-\eps$. Since we have exactly
  $m$ \sts, we have to place one \st, say \st $i$, to close off the
  middle corridor between $x_i$ and $x_{i+1}$.

  Next, observe that if we place \st $i$ to block off the middle corridor
  between $x_i$ and $x_{i+1}-\eps$ at some $x$-coordinate other than $x_i$, we
  can just shift it to $x_i$ without decreasing the length of the shortest
  path. 
  Furthermore note that
  connecting the left endpoint of the bottom gap to the right endpoint of the
  top gap allows either the top or the bottom path to bypass one of the long
  ``spikes'' of length $Q$, hence such a placement is non-optimal.

  \paragraph{Claim: Any shortest path either stays entirely above middle route
    or entirely below it.} We now argue that there are only two potential
  shortest paths left: one that stays entirely above the middle corridor, and
  one that stays below it. Suppose that a shortest path $\pi$
  goes through the $i^\text{th}$ zigzag on the red route, i.e., above the
  corridor, and thus passes through $\ell^+_i$, and then crosses the middle
  route to $r^-_i$. Since we place \st $i$ at $x_i$, this path has to go
  around the horizontal segment $h_i$. It follows that the length of this path
  is at least $X+2H+Y \geq X+2(m+\sum(A))+Y$, where $X=\geod(s,\ell^+_i)$ and
  $Y=\geod(r^-_i,t)$ are the shortest $\s\textrm-\ell^+_i$ and $r^-_i\textrm-t$ paths. Now consider the subpath of $\pi$ from $r^-_i$ to $t$, and
  mirror it in the $x$-axis. Observe that this does not increase its length,
  and that this path $\pi'$ now goes through $r^i_+$. Now consider the path that follows $\pi$ to $\ell^+_i$, then goes to $r^+_i$ via $q^+_i$, and then uses
  $\pi'$ from $r^+_i$ to $t$. This path has length $X+D+a_i+Y$ and is thus
  shorter than $\pi$. That is, $\pi$ cannot be a shortest path. It follows that any
  shortest path either stays entirely above or below the middle corridor.

  \paragraph{Increasing the length of either the red or blue route by
    $a_i$.} Finally, we argue that if we place \st $i$ exactly between $f_i$
  and $p^+_i$ we can increase the length of the shortest path between
  $\ell^+_i$ and $r^+_i$ by an additional $a_i$. Symmetrically, placing \st
  $i$ between $e_i$ and $p^-_i$ instead increases the path between $\ell^-_i$
  and $r^-_i$ by an additional $a_i$.

  It now follows that the length of the ``red'' shortest path that stays
  entirely above the middle corridor and the length of the ``blue'' shortest
  path that stays entirely below both have length (at least) $C+\sum(A)/2$ if
  and only if we can partition $A$ into two subsets $R$ and $B$ with
  $\sum(R)=\sum(B)=\sum(A)/2$. 
\end{proof}


Membership in NP for our problems is open, since verifying solutions involve summing square roots.

\section{Polynomial-time algorithms}\label{sec:poly}

We start from blocking path with \e{one} \st. Recall that two \s-\t paths in the domain have the same \e{homotopy type} if they can be continuously (without intersecting the obstacles) morphed to each other. A \e{locally shortest} (or ``pulled taut'') path is the shortest path within its homotopy type. We consider only those homotopy types for which the locally shortest paths do not self-intersect or self-touch. The shortest \s-\t path \sps{t} is the shortest path of all locally shortest paths. 
It is well known that shortest paths lie on edges of the \e{visibility graph} $VG(\P)$ that connects mutually visible vertices of the domain.

When speaking about a path or a path type, we will assume that the path or the path type exists and is locally shortest: that is, a statement like ``shortest path of type X'' should read ``locally shortest path of type X, if such exists'' (if a path does not exist, its length is set to infinity). We will also assume that any shortest path $\pi$ that we speak about is unique (again, omitting the modality ``if such a path exists''); otherwise, $\pi$ will mean an arbitrary shortest path of the spoken type.

The shortest \s-\t path in the presence of $ab$ is the shorter among the following:\begin{itemize}
\item the shortest path $\pi$ that does not intersect (touch) $ab$ (i.e., the shortest path that goes neither through $a$ nor through $b$)
\item the shortest path(s) intersecting $ab$.\end{itemize}
The latter type of paths can be of the following subtypes (Fig.~\ref{fig:competing}):
\begin{itemize}
\item $\pi_a$ going through $a$ but not $b$
\item $\pi_b$ going through $b$ but not $a$
\item $\pi_{a,b}$ going through both $a$ and $b$. This can either be a path $\pi_{\tilde{ab}}$ that uses more than one edge between $a$ and $b$ (if the \st intersects a hole), or a path $\pi_{ab}$ that has $ab$ as an edge (if $a$ and $b$ are mutually visible).
\end{itemize}
\begin{figure}
\begin{minipage}[t]{.5\columnwidth}
\centering
\includegraphics{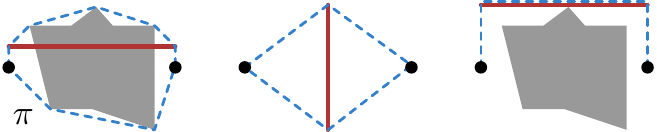}
\caption{The paths (dashed) competing to be the shortest: $\pi$ and $\pi_{\tilde{ab}}$ (left), $\pi_a$ and $\pi_b$ (middle), $\pi_{ab}$ (right); \s and \t are the dots, and $ab$ is bold red.}
\label{fig:competing}
\end{minipage}
\hfill
\begin{minipage}[t]{.46\columnwidth}
\centering
\includegraphics{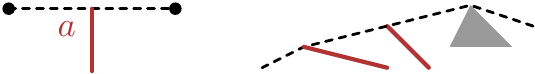}
\caption{Left: $\pi_a$ starting/stopping to be pulled taut implies $a$ moving over the edge of $VG(\P)$ (dashed); red dots are \rs{a} and \rt{b}. Right: formula for path length changes when and edge of $VG(\bar\P)$ appears/disappears.}\label{fig:taut}
\end{minipage}
\end{figure}
The next two lemmas show that the shortest path through the most vital \st cannot be of the last subtype, i.e., that none of $\pi_{\tilde{ab}},\pi_{ab}$ can be an optimal shortest path:\begin{lemma}\label{lem:pi_atildeb}$\pi_{\tilde{ab}}$ cannot be the shortest path for the most vital \st (segment) $ab$.\end{lemma}\begin{proof}
Analogously to Lemma~\ref{lem:touch}, there is always a direction in which $ab$ can be moved so that the lengths of both $\rs{a}a+a\rt{a}$ and $\rs{b}b+b\rt{b}$ increase, lengthening the path.\end{proof}
\begin{lemma}\label{lem:piab_holes}$\pi_{ab}$ cannot be the shortest path for the most vital \st $ab$.\end{lemma}\begin{proof}Analogous to Lemma~\ref{lem:piab}.\end{proof}

Say that two placements of $ab$ are \e{combinatorially equivalent} if in both placements each of $a,b$ lies within the same cell of each of \spms and \spmt, and in both placements the \st intersects the same set of edges of $VG(\P)$; say that the equivalence class defines the \e{combinatorial type} of the \st's placement. For a fixed combinatorial type, it is easy to write lengths of all the competitor paths $\pi,\pi_a,\pi_b$ as functions of $a$ and $b$. Indeed, $\pi$ stays the same, as the same edges of the visibility graph are blocked by $ab$, while $\pi_a$ either always stays pulled taut (and hence its length is $\sps{\rs{a}}+\rs{s}a+a\rt{a}+\spt{\rt{a}}$) or is never pulled taut (and hence is out of the competition) -- this is because when $\pi_a$ starts (resp.\ stops) being pulled taught the visibility edge between \rs{a} and \rt{a} starts (resp.\ stops) being cut by $ab$ (Fig.~\ref {fig:taut}, left); similarly for $\pi_b$.

We scroll through all combinatorial types of the \st placement; since \spms and \spmt have linear complexity and the visibility graph has $O(n^2)$ edges, the number of combinatorial types is polynomial. For each type, we compute the lower envelope of the lengths of the competing paths (the envelope defines the shortest \s-\t path) and take the highest point on the envelope.

We continue to placing a constant number of \sts.

\begin{theorem}\label{thm:path-HbD}
Path-HbD, and hence path-hbD, path-Hb1 and path-hb1, are polynomial.
\end{theorem}
\begin{proof}
If the number of \sts is constant, only $O(1)$ different sets of super-\sts can be formed, and each such set has $O(1)$ super-\sts (since there were  a constant number of \sts to start from). Below we will describe how to deal with one such set; we call the super-\sts in the set just ``\sts''.

Similarly to finding one most vital \st, we say that two placements of the \sts are combinatorially equivalent if each \st endpoint is within the same cells of both \spms and \spmt, and if the same edges of $VG(\bar\P)$ are intersected by the \sts, where $VG(\bar\P)$ is the visibility graph of \P \e{and} the \sts. The number of combinatorially different placements remains polynomial (though exponential in the number of \sts): as the \sts move, $VG(\bar\P)$ changes when three endpoints of the \sts become aligned or when two \st endpoints align with a vertex of \P (Fig.~\ref{fig:taut}, right), which is defined by a polynomial number of constant-description-complexity curves in the constant-dimensional space of \st placements. We again scroll through all combinatorially different placements. For each placement, we compare the lengths of locally shortest simple (non-self-touching) paths of a constant number of homotopy types: a type is defined by the set of \st endpoints touched by the path (together with specifying whether the \st is above or below the path -- since \s and \t are on the outer boundary, the ``above'' and ``below'' are well defined) and the order in which the endpoints are touched -- altogether these define the homotopy type uniquely \cite[Section~4]{socg07}. The rest is the same as with placing one \st: build the lower envelope of the lengths of the (locally shortest) paths and take the highest point.
\end{proof}


\subsection{Polynomial-time algorithms for flow blocking}\label{sec:flow}

By the geometric MaxFlow/MinCut Theorem and the fact that MinCut is the shortest \B-\T path in the critical graph $G$ (see Section~\ref{sec:prelim}), decreasing the flow is equivalent to decreasing the length of the shortest path in $G$.\begin{lemma}\label{lem:G}There exists an optimal solution where \sts are placed along edges of $G$.\end{lemma}\begin{proof}Suppose a subpath of a \B-\T path, going between holes $R$ and $Q$, deviates from edges of $G$ (Fig.~\ref{fig:shortcut}). Let $D$ be the length that the subpath spends inside \sts, and let $d$ be the length of the subpath in the free space (it is possible that either of $D,d$ is 0). By the triangle inequality, the distance between the holes $R$ and $Q$ (i.e., the length of the edge $RQ$ of the critical graph) is at most $D+d$. Moving the \sts from the subpath onto the edge shortens its length by at least $D$, decreasing the overall length of the path (recall from Section~\ref{sec:tax} that we allow \sts to intersect holes, so we may freely move the \sts).\end{proof}
The above lemma reduces flow blocking to a purely graph-theoretic problem: shorten the \B-\T path as much as possible using the given \sts, where a length-$l$ \st shortens a length-$L$ edge to $\max[0,L-l]$.

\begin{figure}
\centering
\includegraphics{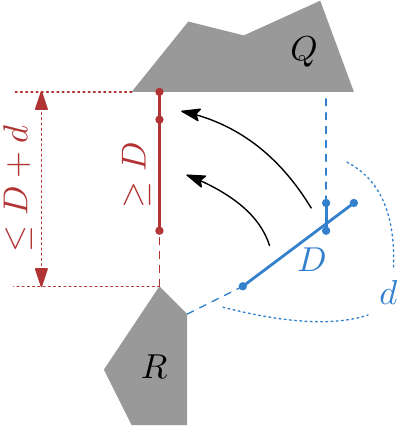}
\caption{\Sts are bars with dots at endpoints. A subpath of a shortest \B-\T path, going between two holes and not using edges of $G$ (blue), can be replaced by a shorter path (red) with the same \sts aligned along the edge of $G$.}\label{fig:shortcut}
\end{figure}
\begin{theorem}
\label{thm:flow-HB1}
Flow-HB1 can be solved in pseudopolynomial time.
\end{theorem}
\begin{proof}Let $K$ be the number of \sts. Similarly to the standard pseudopolynomial-time algorithms for bi-criteria shortest paths in graphs with two kinds of edge lengths (sometimes called weights and costs) \cite{piatko}, we propagate $K$ labels from \B; label $k=1\dots K$ of a vertex is the length of the shortest path from \B to the vertex, using $k$ \sts. When propagating a label $l_k(u)$ along an edge $uv$ of length $L$, every label $i\ge k$ of $v$ is updated to the minimum of its current $i$th label $l_i(v)$ and $l_k(v)+\max[0,L-(i-k)]$, signifying the placement of $i-k$ \sts along $uv$.
\end{proof}
For a constant number of \sts we have:
\begin{theorem}
\label{thm:flow-HbD}
Flow-HbD, and hence flow-hbD, flow-Hb1 and flow-hb1, are polynomial.
\end{theorem}
\begin{proof}Analogously to Theorem~\ref{thm:path-HbD}, only $O(1)$ super-\sts can be formed from a constant number of \sts. For each (constant-size) set of the super-\sts, there is a polynomial number of ways to place the super-\sts on edges of the critical graph (since $G$ has $O(h^2)=O(n^2)$ edges where $h$ is the number of holes in \P); for each such placement the shortest \B-\T path is computed, and the best overall placement is chosen.
\end{proof}
The remaining problem is flow-hB1:
\begin{theorem}
\label{thm:flow-hB1}
Flow-hB1 is polynomial.
\end{theorem}
\begin{proof}With $O(1)$ holes, there are only $O(1)$ \B-\T paths. For each \B-\T path, we place the \sts greedily: First, we place $\lfloor L\rfloor$ \sts on each edge of length $L$ of the path (until we run out of the \sts) -- this way we do not waste the \sts. Then, if any \sts remain and there is an edge of the path not yet fully covered by the \sts, we place one more \st per edge, in decreasing order of the length of the part of the edge which is not yet covered. That is, we first place a \st on the edge with the largest fractional part of the length (eating up as much of the remaining path length as possible), then on the edge with the second largest fractional part, etc. (this is optimal, since we waste as little total \st length as possible).
\end{proof}

\section{Conclusion}We introduced geometric versions of the graph-theoretic most vital arcs problem. 
We presented efficient solutions for simple polygons, and gave hardness results and algorithms for various versions of the problem. The most intriguing open problem is the hardness of path-hB1 (path blocking with few holes, our only unresolved version); we conjecture that it is polynomial, as still only a constant-number of super-\sts may be needed. Another interesting question is whether the flow and the path blocking have fundamentally different complexities: we proved that the complexities are the same for all versions except HB1 -- for path-HB1 we showed weak hardness but lack a pseudopolynomial-time algorithm, while for flow-HB1 we have a a pseudopolynomial-time algorithm but no (weak) hardness proof.
\iffun
\else
More generally, various other setups may be considered. For instance, one may be given a budget on the total length of the \sts -- the problem then is how to split the budget between the \sts and where to locate them. For minimizing the maximum flow this version is easy: just place the \sts along the shortest \B-\T path in the critical graph of the domain. For maximizing the shortest path in a simple polygon the solution is trivial: make a single \st of the full length (and use our algorithm to find the optimal \st location). Blocking shortest paths in polygons with holes in this setting is an open problem. \val{do we wanna mention a mazing problem (including the possible hardness proof for polygons with holes)?}
\fi

\bibliographystyle{abbrv}
\bibliography{document}
\end{document}